\documentclass[12pt]{article}

\usepackage[T1]{fontenc}
\usepackage{times}
\usepackage{fullpage}
\usepackage{amsfonts,amssymb,amsmath,amsthm}
\usepackage{latexsym}
\usepackage{gauss}
\usepackage{calc}
\usepackage{url}
\usepackage{thmtools}
\usepackage{thm-restate}
\usepackage{hyperref}
\usepackage[capitalise, nameinlink]{cleveref}
\usepackage{algorithm}
\usepackage{algorithmicx}
\usepackage{algpseudocode}
\usepackage{graphicx}
\usepackage{enumitem}
\usepackage{tikz}
\newcommand{\N}{\mathbb{N}}
\newcommand{\R}{\mathbb{R}}
\newcommand{\Pcal}{\mathcal{P}}
\newcommand{\tO}{\widetilde O}
\newcommand{\eps}{\varepsilon}

\newcommand{\st}{s{\operatorname{-}}t}
\newcommand{\ab}{a{\operatorname{-}}b}

\newcommand{\ith}{i^{\scriptsize \mbox{{\rm th}}}}

\newtheorem{theorem}{Theorem}

\newtheorem{lemma}[theorem]{Lemma}
\newtheorem{corollary}[theorem]{Corollary}

\newtheorem{remark}[theorem]{Remark}

\newtheorem{claim}[theorem]{Claim}

\theoremstyle{definition}

\newtheorem{definition}[theorem]{Definition}

\begin{document}
\title{A sublinear query quantum algorithm for $s{\operatorname{-}}t$ minimum 
cut on dense simple graphs}
\author{Simon Apers\thanks{IRIF, CNRS, Paris. Email: smgapers@gmail.com} \and 
Arinta Auza\thanks{Centre for Quantum Software and Information, University of 
Technology Sydney.} \and Troy Lee\thanks{Centre for Quantum Software and 
Information, University of Technology Sydney. Email: troyjlee@gmail.com}}
\date{}
\maketitle

\begin{abstract} 
An $s{\operatorname{-}}t$ minimum cut in a graph corresponds to a minimum 
weight subset of edges whose removal disconnects vertices $s$ and $t$.
Finding such a cut is a classic problem that is dual to that of finding a 
maximum flow from $s$ to $t$.  In this work we describe a quantum algorithm 
for the minimum $s{\operatorname{-}}t$ cut problem on undirected graphs.
For an undirected graph with $n$ vertices, $m$ edges, and integral edge 
weights bounded by $W$, the algorithm computes with high probability the 
weight of a minimum $s{\operatorname{-}}t$ cut after 
$\widetilde O(\sqrt{m} n^{5/6} W^{1/3})$ queries to the adjacency list 
of $G$.  For simple graphs this bound is always $\widetilde O(n^{11/6})$, 
even in the dense case when $m = \Omega(n^2)$.  In contrast, a randomized 
algorithm must make $\Omega(m)$ queries to the adjacency list of a simple 
graph $G$ even to decide whether $s$ and $t$ are connected.
\end{abstract}

\section{Introduction}
Let $G$ be a graph with $n$ vertices and $m$ edges, and let $s,t$ be two 
vertices of $G$.  The problem of determining the maximum amount of flow 
$\lambda_{st}(G)$ that can be routed from $s$ to $t$ while 
respecting the capacity constraints of $G$ is one of the most fundamental 
problems in theoretical computer science, whose study goes back at least to 
the 1950s and the pioneering work of Ford and Fulkerson \cite{FF56}.
By the max-flow min-cut theorem, $\lambda_{st}(G)$ is equal to the minimum 
weight of a cut separating $s$ and $t$ in $G$.  From a maximum $\st$ flow 
one can compute a minimum $\st$ cut in linear time, but no such reduction 
is known the other way around.  The Goldberg-Rao \cite{GR98} algorithm, 
with running time $O(\min\{\sqrt{m}, n^{2/3}\} m \log(n) \log(W))$, where 
$W$ is the largest weight of an edge, stood as the best bound for both 
problems for many years.  In the past decade, however, beginning with work of 
Christiano, Kelner, M\k{a}dry, Spielman, and Teng  \cite{CKMST11} there has 
been tremendous progress in 
max-flow algorithms by incorporating techniques from continuous optimization 
\cite{Sherman13, Madry13, KLOS14, Peng16, LS19, LS20, GLP21, BLLSSSW21}.  
With a recent $m^{1 + o(1)} \log W$ time max-flow algorithm by 
Chen, Kyng, Liu, Peng, Gutenberg, and Sachdeva \cite{CKLPGS22}, there are now 
almost-linear time randomized 
max-flow algorithms for graphs with polynomially bounded integral weights.  

Given its fundamental nature, there has been a surprising lack of work on 
quantum algorithms for the exact maximum flow or minimum $\st$ cut problem.  
As far as we are aware, the only work on the quantum complexity of these 
problems is by Ambainis and \v{S}palek \cite{AS05}, who gave a quantum 
algorithm for max flow in a directed graph with integral weights bounded 
by $W \le n^{1/4}$ with running time 
$\tO(\min\{n^{7/6} \sqrt{m} W^{1/3}, m\sqrt{nW}\})$, given adjacency list 
access to $G$.  This bound is completely subsumed by current classical 
randomized algorithms.  

More recent work of Apers and de Wolf \cite{AdW19} gives a 
$\tO(\sqrt{mn}/\eps)$ time quantum algorithm for finding a 
$(1 + \eps)$-approximation of the minimum $\st$ cut.\footnote{The bound quoted 
in \cite[Claim 9]{AdW19} for finding a 
$(1+\eps)$-approximate minimum $\st$ cut is $\tO(\sqrt{mn}/\eps + n/\eps^5)$.  
However, plugging into the proof the improved randomized approximate $\st$ 
minimum cut algorithm of \cite{ST18} with running time 
$\tO(m + \sqrt{mn}/\eps)$ improves this to $\tO(\sqrt{mn}/\eps)$.}  
This algorithm works by first constructing an $\eps$-cut sparsifier $H$ of the 
input graph $G$ with a quantum algorithm in time $\tO(\sqrt{mn}/\eps)$.  
An $\eps$-cut sparsifier is a re-weighted subgraph of $G$ that has only 
$\tO(n/\eps^2)$ edges but for which the weight of every cut agrees with that 
of $G$ up to a factor of $1\pm \eps$.  One can then run a classical randomized 
$\st$ minimum cut algorithm on the sparse graph $H$ to find an approximate 
$\st$ minimum cut of $G$.  The quantum sparsification algorithm also plays a 
key role in our work.

In this paper, we give a quantum algorithm that computes $\lambda_{st}(G)$ on 
dense simple graphs with fewer queries than is possible classically.\footnote{A 
simple graph is an undirected and unweighted graph with no self 
loops and at most one edge between any pair of vertices.} More generally, 
for an undirected and integral weighted graph $G$ with maximum weight $W$, 
we give a quantum algorithm with adjacency list access to $G$ that with high 
probability outputs $\lambda_{st}(G)$ and makes 
$\tO(\sqrt{m} n^{5/6} W^{1/3})$ queries (see \cref{thm:stmincut}).  
Thus for unweighted graphs the number of queries is $\tO(n^{11/6})$ 
even for dense instances with $m \in \Omega(n^2)$ edges.  On the other hand, 
it is easy to show that in the worst case a randomized algorithm requires 
$\Omega(m)$ queries to the adjacency list of a graph with $m$ edges even to 
decide whether $s$ and $t$ are connected (see \cref{thm:rand_lower}).  
Our algorithm also works with adjacency matrix access to $G$, in which case 
the number of queries becomes $\tO(n^{11/6}W^{1/3})$ 
(see \cref{cor:adj}).  

In addition to $\lambda_{st}(G)$, our algorithm can 
output the corresponding bipartition of the vertex set.  It does not, however, 
compute a maximum $\st$ flow, and finding a sublinear quantum query algorithm 
for this problem remains a tantalizing open question.

Our quantum algorithm follows an algorithm of Rubinstein, Schramm and 
Weinberg (RSW) \cite{RSW18}, which computes a minimum $\st$ cut in the 
\emph{cut query} model.  In the cut query model one can query a subset of 
vertices $S$ and receive as an answer the total weight of edges with 
exactly one endpoint in $S$.  While designed for cut query complexity, 
we show the algorithm also has a surprisingly efficient implementation 
in the adjacency list and adjacency matrix models that is well-suited for 
quantum speedups.  The crux of the algorithm is a procedure to transform the 
input graph $G$ into a sparser graph $G'$ while preserving minimum $\st$ cuts.  
We use two quantum tools to speed up this procedure.  The first is to use 
the quantum algorithm from \cite{AdW19} to find an $\eps$-cut sparsifier 
of $G$ faster than a classical sparsification 
algorithm.  The second is to use Grover's algorithm to learn all the edges 
in the sparser graph $G'$.  Once we explicitly know the graph $G'$ we can 
use a classical randomized algorithm to compute a minimum $\st$ cut in $G'$.
This blueprint is similar to the global min cut algorithm from \cite{AL20}, 
which is also based on a (different) cut query algorithm from \cite{RSW18}.  
There as well a cut sparsifier of $G$ is used to identify edges that can 
be contracted while preserving (non-trivial) global minimum cuts, and then 
Grover search is applied to learn the edges in the sparser contracted graph.

The contracted graph in the RSW procedure is obtained by contracting edges 
that cannot participate in any minimum $\st$ cut.  To get an intuition for 
the idea, suppose that $G$ is a simple graph and first imagine that we 
compute a maximum $\st$ flow $F$ in $G$.  If we let $G_F$ be the graph
whose edge weights are those of $G$ minus the flow $F$, then $s$ and $t$ 
become disconnected in $G_F$---otherwise we could route more flow 
from $s$ to $t$.  Contracting all connected components of $G_F$ then 
results in a sparser graph while preserving any edge that is part of a 
minimum $\st$ cut.  This routine is also how one can compute a 
minimum $\st$ cut from a maximum $\st$ flow: the connected component of $s$ 
in $G_F$ gives one side of a vertex bipartition corresponding to a minimum 
$\st$ cut.

Of course we did not save anything in this example as we had to compute a 
maximum $\st$ flow in $G$.  What RSW actually do is to first obtain an 
$\eps$-cut sparsifier $H$ of $G$.  This is where we use the 
$\tO(\sqrt{mn}/\eps)$ time quantum sparsification algorithm from \cite{AdW19}.
As $H$ has only $\tO(n/\eps^2)$ edges it is much cheaper to carry out the 
above plan on $H$ instead: we compute a maximum flow $F$ in $H$ and consider 
the graph $H_F$ whose weights are those of $H$ minus the flow $F$.  As 
cuts of $H$ only approximate those in $G$ we cannot fully contract the 
connected components in $H_F$ and hope to preserve all minimum $\st$ cut of 
$G$. However, a key insight of RSW is that we can still safely contract 
induced subgraphs of $H_F$ that are highly connected.  Doing so results in 
the desired contraction $G'$ of $G$ that is relatively sparse---one can 
argue it has only $O(\eps n^2)$ edges---yet preserves all minimum $\st$ cuts.  
We can learn the $O(\eps n^2)$ edges of $G'$ among the $m$ edges of $G$ using 
Grover search in time $\tO(n \sqrt{m \eps})$.  We then again run a 
max flow algorithm on the now explicitly known $G'$ to compute 
$\lambda_{st}(G)$.  Balancing the terms $\sqrt{mn}/\eps$ from the sparsifier 
computation and $n \sqrt{m \eps}$ for learning the edges of $G'$ leads to the 
choice $\eps = n^{-1/3}$, and $\tO(\sqrt{m} n^{5/6})$ queries 
for the quantum steps of the algorithm.  All other steps are on explicitly
learned graphs and require no queries.

\section{Preliminaries}
\subsection{Graph notation and background}
Let $V$ be a finite set and $V^{(2)}$ the set of all subsets of $V$ of 
cardinality $2$.  We represent a weighted undirected graph as a pair 
$G = (V,w)$ where $w: V^{(2)} \rightarrow \R$ is a non-negative function.  
We let $V(G)$ be the vertex set of $G$ and 
$E(G) = \{ e \in V^{(2)}: w(e) > 0\}$ be the set of edges of $G$.  We extend 
the weight function to sets 
$S \subseteq V^{(2)}$ by $w(S) = \sum_{e \in S} w(e)$.
We say that $G$ is \emph{simple} if $w: V^{(2)} \rightarrow \{0,1\}$ and in 
this case also denote $G$ as $G = (V,E)$, where $E$ is the set of edges.  
For a subset $X \subseteq V$ we write $G[X]$ for the induced subgraph on $X$.

For a subset $X \subseteq V$ we use the shorthand 
$\overline{X} = V \setminus X$, and we say $X$ is \emph{non-trivial} if 
$\emptyset \ne X \subsetneq V$.  For disjoint sets $X,Y \subseteq V$ we use 
$E(X,Y)$ for the set of edges with one endpoint in $X$ and one endpoint in $Y$.
For a non-trivial set $X$, let 
$\Delta_G(X) = \{ \{i,j\} \in E(G) : i \in X, j \in \overline{X}\}$ be the 
set of edges of $G$ with one endpoint in $X$ and one endpoint in 
$\overline{X}$.  A \emph{cut} of $G$ is a set of the form $\Delta_G(X)$ for 
some non-trivial set $X$.  We call $X$ and $\overline{X}$ the \emph{shores} 
of the cut $\Delta_G(X)$.  For two distinguished vertices $s,t \in V$ we 
call $\Delta_G(X)$ an $\st$ cut if $s \in X, t \not \in X$.  
When we speak of the shore of an $\st$ cut we always refer to the shore 
containing $s$.  We let $\lambda_{st}(G)$ be the minimum weight of an 
$\st$ cut of $G$.

By the famous max-flow min-cut theorem \cite{FF56}, $\lambda_{st}(G)$ is 
equal to the maximum amount of flow that can be routed from $s$ to $t$ in $G$.  
While all the graphs $G$ in this paper will be undirected, for flows it is 
most natural to think of an undirected graph $G$ as a directed graph with 
each edge directed in both directions.  Without loss of generality, an 
$\st$ flow $F$ in $G$ can be defined to use an edge in at most one direction.  
For $e \in E(G)$ it will be convenient for us to let $F(e)$ be the flow along 
edge $e$ in the direction it is used by the flow $F$.

We will need the following result about the total weight of a flow from 
Rubinstein, Schramm, and Weinberg \cite{RSW18}.
\begin{lemma}[{\cite[Lemma 5.4]{RSW18}}]
\label{lem:rsw_flow}
Let $G = (V,w)$ be a graph with integral weights from $[0,W]$ and let 
$s,t \in V$.  Let $F$ be a non-circular $\st$ flow of value $f$ in $G$.  
Then the total weight of flow $\sum_{e \in E(G)} F(e) \le 10\cdot n\sqrt{fW}$.
\end{lemma}

\subsection{Quantum models and background}
We will work with the two most standard quantum models for graph algorithms, 
the adjacency list and adjacency matrix models.  We refer the reader to 
\cite[Section 2.3]{AL20} for definitions of these models.  

Our quantum algorithm can be viewed as a classical algorithm with two quantum 
primitives.  The first is a quantum algorithm to compute a sparsifier 
\cite{AdW19} described in the next section, and the second is the following 
application of Grover's algorithm.
\begin{theorem}[cf.\ {\cite[Theorem 13]{AL20}}]
\label{thm:qsearch}
Given $t,N \in \N$ with $1 \le t \le N$ and oracle access to $x \in \{0,1\}^N$,  
there is a quantum algorithm such that
\begin{itemize}
  \item if $|x| \le t$ then the algorithm outputs $x$ with certainty, and
  \item if $|x| > t$ then the algorithm reports so with probability at least 
$9/10$.
\end{itemize}
The algorithm makes $O(\sqrt{tN})$ queries to $x$ and has time complexity 
$O(\sqrt{tN} \log(N))$.
\end{theorem}

\subsection{Sparsifiers and strong components}
A key component of our algorithm will be a cut sparsifier of a graph.
\begin{definition}[cut sparsifier]
Let $G = (V,w_G)$ be a weighted graph.  An $\eps$-cut sparsifier 
$H = (V,w_H)$ of $G$ is a reweighted subgraph of $G$, i.e., $w_H(e) > 0$ only 
if $w_G(e) > 0$, satisfying 
\[
(1-\eps) w_G(\Delta_G(X)) \le w_H(\Delta_H(X)) \le (1+\eps) w_G(\Delta_G(X))
\]
for every $X \subseteq V$.
\end{definition}

We will use a quantum algorithm to construct a cut sparsifier of a graph by 
Apers and de Wolf \cite{AdW19}.
\begin{theorem}[{\cite[Theorem 1]{AdW19}}]
\label{thm:AdW}
Let $G$ be a weighted and undirected graph with $n$ vertices and $m$ edges.
There exists a quantum algorithm that outputs with high probability the 
explicit description of an $\eps$-cut sparsifier $H$ of $G$ with 
$\tO(n/\eps^2)$ edges after $\tO(\sqrt{mn}/\eps)$ queries to the adjacency list 
model or $\tO(n^{3/2}/\eps)$ queries in the adjacency matrix model.  
Moreover, if $G$ has integral weights then so does $H$, and if the largest 
weight of an edge of $G$ is $W$ then the largest weight of an edge of $H$ is 
$\tO(\eps^2 n W)$.
\end{theorem}

\begin{proof}
The first part of the statement is directly from Theorem~1 of \cite{AdW19}.  
The ``moreover'' part follows from an examination of Algorithm 1 of 
\cite{AdW19}.  There it can be seen that when an 
edge is added to the sparsifier it is always done so with a power of $4$ 
times its original weight.  The power of $4$ is at most the number of times 
Algorithm 1 is run, which depends on the error parameter $\eps$.  To achieve 
an $\eps$-cut sparsifier, Algorithm 1 is run 
$k=\log(n \eps^2)/2 + O(\log \log n)$ times, thus an edge of $G$ is placed in 
the sparsifier with weight at most its original weight times 
$4^k = \tO(\eps^2 n)$. 
\end{proof}

Apers and de Wolf actually show how to construct a \emph{spectral sparsifier}. 
However, we will not need the stronger properties of a spectral sparsifier.

We will also need some definitions and results related to the minimum weight of cuts in induced subgraphs of $G$.

\begin{definition}[$k$-strong component]
A graph $G$ is $k$-connected if the weight of each cut in $G$ is at least $k$.
A $k$-strong component of $G$ is a maximal $k$-connected vertex induced subgraph of $G$.  Individual vertices are 
defined to be $\infty$-strong components.
\end{definition}

\begin{definition}[edge strength]
The strength of an edge $e$, denoted $k_e$ is the maximum value of $k$ such that a $k$-strong component 
contains both endpoints of $e$.  We say that $e$ is $k$-strong if its strength is $k$ or more, and $k$-weak 
otherwise.
\end{definition}

Bencz\'{u}r and Karger \cite{BK15} show the following lemma about edge strengths.
\begin{lemma}[Lemma 4.11 \cite{BK15}]
\label{lem:BK}
Let $G = (V,w)$ be an $n$-vertex graph where the strength of edge $e$ is $k_e$. 
Then
\[
\sum_e \frac{w(e)}{k_e} \le n-1 \enspace.
\]
\end{lemma}

\begin{definition}[$k$-strong partition]
\label{def:strong_partition}
For a graph $G = (V,w)$, a $k$-strong partition of $G$ is a partition $\Pcal = \{S_1, \ldots, S_t\}$ of $V$ such that 
each $G[S_i]$ is a $k$-strong component for $i=1,\ldots, t$.
\end{definition}

Bencz\'{u}r and Karger give an algorithm to find a $k$-strong partition in
Lemma 4.8 of \cite{BK15}.  If the weight of a minimum cut of $G$ is at least 
$k$ then $\{V\}$ provides a $k$-strong partition.  Otherwise one finds 
\emph{all} minimum cuts of $G$ and removes the union of all of their edges.  
All of these edges have edge strength less than $k$, and by Lemma 4.5 of 
\cite{BK15} removing a $k$-weak edge does not affect the edge strength of a 
$k$-strong edge.  One then applies the same procedure to the components 
created after the removal of these edges.  While not particularly fast, this
provides an algorithm to compute a $k$-strong partition.  

\begin{remark}
The first arXiv version of this paper claimed (in Corollary 10) a nearly-linear
time algorithm to compute a $k$-strong partition.  As pointed out by an 
ICALP 2022 reviewer, the proof of this had a fatal flaw, and we currently 
do not know a nearly-linear time algorithm to compute a $k$-strong partition. 
This is the bottleneck to making our quantum algorithm for $\st$-mincut time 
efficient instead of just query efficient.
\end{remark}

A key property about a $k$-strong partition that makes $\st$-mincut algorithm 
of Rubinstein, Schramm, and Weinberg \cite[Algorithm 5.1]{RSW18}
work efficiently is the following.
\begin{theorem}
\label{thm:partition_edges}
For a graph $G = (V,w)$, let $\{S_1, \ldots, S_t\}$ be a $k$-strong partition.
Then 
\[
\sum_{{i,j} \in [t]^{(2)}} w(S_i, S_j) \le k (t-1) \enspace .
\]
\end{theorem}

\begin{proof}
Each $G[S_i]$ is a $k$-strong component, thus any edge with both endpoints 
in $S_i$ is $k$-strong.  Also, any edge with endpoints in distinct sets $S_i, S_j$ 
is $k$-weak.  By Lemma 4.6 of \cite{BK15}, contracting a $k$-strong 
edge does not change the strength of any $k$-weak edge.  Consider the graph 
$G'$ formed by contracting each $S_i$.  Every edge of $G'$ is $k$-weak.  By 
\cref{lem:BK} the total edge weight of $G'$ is at most $k(t-1)$, which gives 
the lemma.
\end{proof}

\section{Main algorithm} \label{sec:main-algo}
Our algorithm is based on the following algorithm by Rubinstein, Schramm, and 
Weinberg \cite[Algorithm 5.1]{RSW18}, who used it to give a randomized cut 
query algorithm for the minimum $\st$ cut problem making $\widetilde O(n^{5/3})$ 
cut queries. 

\begin{algorithm}[H]
\caption{Algorithm for $\st$-mincut on a weighted graph $G$}
\label{alg:stmincut}
 \hspace*{\algorithmicindent} \textbf{Input:} Oracle access to a weighted graph $G =(V,w_G)$, vertices $s,t \in V$, and a parameter $0 < \eps < 1/3$. \\
 \hspace*{\algorithmicindent} \textbf{Output:} The shore of a minimum $\st$ cut in $G$ and the value $\lambda_{st}(G)$.
\begin{algorithmic}[1]
\State Determine the largest edge weight $W$ in $G$.
\label{step:max}
\State Compute an $\eps$-cut sparsifier $H$ of $G$. 
\label{step:sparsifier}
\State Compute a maximum $\st$ flow $F$ in $H$ and subtract $F$ from $H$.  Denote the result as $H'$. 
\label{step:maxflow}
\State Compute a $3\eps nW$-strong partition $\Pcal = \{A_1, \ldots, A_t\}$ of $H'$. 
\label{step:contract}
\State Let $G'$ be the graph formed from $G$ by contracting the vertices in each $A_i$, and let $a,b \in \Pcal$ be the 
sets containing $s,t$ respectively.  Compute a minimum $\ab$ cut $\Delta_{G'}(X')$ in $G'$ and return $\lambda_{ab}(G')$ and 
the shore $X = \cup_{A \in X'} A$.
\label{step:learnedges}
\end{algorithmic}
\end{algorithm}

For completeness we show correctness of the template, largely following the discussion of \cite{RSW18}.
In the next section we analyze the running time for a quantum algorithm that implements this algorithm with adjacency list access to $G$.

\begin{theorem}
\label{thm:correct}
If every step of \cref{alg:stmincut} is performed correctly then the algorithm returns a minimum $\st$ cut of $G$.
\end{theorem}

We first prove two claims from which the proof of \cref{thm:correct} will easily follow.
\begin{claim}
\label{clm:nocut}
Let $G = (V,w_G)$ be an $n$-vertex weighted graph with largest edge weight 
$W$ and $s,t \in V$.  For $0< \eps < 1/3$, let $H = (V,w_H)$ be an $\eps$-cut 
sparsifier of $G$ and let $F$ be a maximum $\st$ flow in $H$.  Form the 
graph $H' = (V,w_{H'})$ where $w_{H'}(e) = w_H(e) - F(e)$ for all 
$e \in V^{(2)}$.  Let $\Pcal$ be a $3 \eps nW$-strong partition of $H'$.
Then for any minimum $\st$ cut $\Delta_G(X)$ of $G$ and all $A \in \Pcal$ it 
holds that either $A \subseteq X$ or $A \subseteq \overline{X}$.
\end{claim}

\begin{proof}
Let $f_G = \lambda_{st}(G)$ and $f_H = \lambda_{st}(H)$.  
Let $\Delta_G(X)$ be a minimum $\st$ cut of $G$, i.e., such that 
$s \in X, t \not \in X$ and $w_G(\Delta_G(X)) = f_G$.  We want to 
upper bound $\tau = w_{H'}(\Delta_{H'}(X))$.
By definition of $H'$ and $F$, 
$w_H(\Delta_H(X)) = \tau + w_F(\Delta_F(X))$.
We must have $w_F(\Delta_F(X)) \ge f_H$ 
since there is a flow of value $f_H$ from $s$ to $t$ in $H$.  Thus 
$w_H(\Delta_H(X)) \ge f_H + \tau$.  As $H$ is an
$\eps$-cut sparsifier of $G$, $w_H(\Delta_H(X)) \le (1+\eps) f_G$ and so 
$\tau \le (1+\eps) f_G - f_H$.  

If $f_H \ge f_G$, then we immediately have 
$\tau \le \eps f_H \le \eps(1+\eps) nW$ as $f_H \le (1+\eps) f_G < 
(1+\eps) nW$.  In the case $f_H < f_G$, we use the lower bound 
$f_H \ge f_G/(1-\eps)$ to obtain $\tau \le \frac{2\eps}{1-\eps} f_H \le 
\frac{2\eps}{1-\eps}nW$ as $f_H < f_G < nW$.  In either case we have
$\tau < 3\eps nW$ as $\eps < 1/3$. 

Now let $A \in \Pcal$ and let $B = A \cap X$.  If $B$ is nontrivial, 
i.e., $\emptyset \neq B \subsetneq A$, then there is a cut of $H'[A]$ of weight 
less than $3 \eps nW$, a contradiction to the assumption that $\Pcal$ is 
a $3 \eps nW$-strong partition of $H'$.  Thus $B$ must be trivial and 
either $A \subseteq X$ or $A \subseteq \overline{X}$.
\end{proof}

\begin{claim}
\label{clm:cuts_equal}
Let $G,H,H'$ and $\eps$ be as in \cref{clm:nocut}.  Let $\Pcal = \{A_1, \ldots, A_t\}$ be a $3 \eps nW$-strong partition of $H'$ and $G'$ be the graph 
formed from $G$ by contracting vertices in the same set of $\Pcal$.  Let $a,b \in \Pcal$ be the sets containing $s$ and $t$ respectively.  Then 
$\lambda_{st}(G) = \lambda_{ab}(G')$.  Moreover, if $X'$ is the shore of a minimum $\ab$ cut of $G'$ then $X = \cup_{A \in X'} A$ is the shore of a 
minimum $\st$ cut of $G$.
\end{claim}

\begin{proof}
Let $G = (V, w_G)$ and note that by the definition of $G'$ we have $G' = (\Pcal,w_{G'})$ 
where $w_{G'}(A_i,A_j) = w_G(E(A_i,A_j))$ for $i \ne j$. 

Let us show that $\lambda_{ab}(G') \le \lambda_{st}(G)$.  Let $\Delta_G(X)$ be a minimum $\st$ cut in $G$.  For every $A_i \in \Pcal$ 
we either have $A_i \subseteq X$ or $A_i \subseteq \overline{X}$ by \cref{clm:nocut}.  Note that in particular this means $a \subseteq X$ and 
$b \cap X = \emptyset$.  From $X$ we define a set $X'$ where 
for every $A_i \in \Pcal$ we put $A_i$ in $X'$ iff $A_i \subseteq X$.  Note that $\Delta_{G'}(X')$ is an $\ab$ cut of $G'$ and 
$w_G(\Delta_G(X)) = w_{G'}(\Delta_{G'}(X'))$, thus $\lambda_{ab}(G') \le \lambda_{st}(G)$.

We next show the general fact that contraction cannot decrease the minimum $\st$ cut value.  For any set $X' \subseteq \Pcal$ 
we can define a set $X \subseteq V$ by $X = \cup_{A \in X'} A$.  
Further $w_{G'}(\Delta_{G'}(X')) = w_G(\Delta_G(X))$ by the definition of $w_{G'}$.  
Thus $\lambda_{ab}(G') \ge \lambda_{st}(G)$.  

This establishes $\lambda_{st}(G)=\lambda_{ab}(G')$.  To see the ``moreover'' part of the claim, let $X'$ be a minimum $\ab$ 
cut of $G'$.  We have just shown that $w_{G'}(\Delta_{G'}(X')) = w_G(\Delta_G(X))$ for $X = \cup_{A \in X'} A$, thus $X$ will be a 
minimum $\st$ cut of $G$.
\end{proof} 

\begin{proof}[Proof of \cref{thm:correct}]
If the steps of the algorithm are implemented correctly then $H,H'$ satisfy the hypotheses of \cref{clm:nocut}.  We can therefore invoke \cref{clm:nocut} and 
\cref{clm:cuts_equal} to conclude that $\lambda_{st}(G) = \lambda_{ab}(G')$ and that $X$ will be the shore of a minimum $\st$ cut in $G$.
\end{proof}

\subsection{Implementation by a quantum algorithm}
We now discuss the implementation of \cref{thm:stmincut} by a quantum 
algorithm with adjacency list access to $G$.  To upper bound the quantum query 
complexity it is important to have an upper bound on the number of edges in 
the contracted graph $G'$ formed in step~\ref{step:learnedges} of 
\cref{alg:stmincut}.  We do this by means of the following claim.
\begin{claim}
\label{clm:weight_upper}
Let $G,H,H'$ and $\eps$ be as in \cref{clm:nocut}, with the following additional conditions:
\begin{enumerate}
\item The largest edge weight of $G$ is $W \ge 1$.
\item $H$ has integral weights and the largest edge weight is $W_H \in O(\eps^2 n W)$.
\end{enumerate}
Let $\Pcal = \{A_1, \ldots, A_t\}$ be a $3 \eps nW$-strong partition of $H'$.  
Then $\sum_{i=1}^t w_G(\Delta_G(A_i)) = O(\eps n^2 W)$.
\end{claim}

\begin{proof}
As $H$ is an $\eps$-cut sparsifier of $G$ we have
\begin{align*}
w_G(\Delta_G(A_i)) &\le  \frac{w_H(\Delta_H(A_i))}{1-\eps} \\
&=  \frac{w_{H'}(\Delta_{H'}(A_i)) + w_F(\Delta_F(A_i))}{1-\eps} \enspace .
\end{align*}
By \cref{thm:partition_edges}, we have 
$\sum_{i=1}^t  w_{H'}(\Delta_{H'}(A_i)) = O(\eps n^2W)$.
By \cref{lem:rsw_flow} we have that 
\begin{align*}
\sum_{i=1}^t w_F(\Delta_F(A_i)) &\le 20 n \sqrt{\lambda_{st}(H) W_H} \\ 
&= O(\eps n^2 W) \enspace ,
\end{align*}
using that $\lambda_{st}(H) \le (1+\eps) nW$ and $W_H \in O(\eps^2 n W)$.  
Thus overall $\sum_{i=1}^t w_G(\Delta_G(A_i)) = O(\eps n^2 W)$.
\end{proof}

\begin{theorem}
\label{thm:stmincut}
Let $G$ be a graph with $n$ vertices and $m$ edges where all edge weights are 
integral and the largest weight of an edge is $W$.
Given adjacency list access to $G$, there is a quantum algorithm that with 
high probability computes $\lambda_{st}(G)$ and the shore of a 
minimum $\st$ cut of $G$ after $\tO(\sqrt{m} n^{5/6} W^{1/3})$ queries.
\end{theorem}

\begin{proof}
We follow the algorithm given in \cref{alg:stmincut} with the choice 
$\eps = (nW)^{-1/3}$.  

\paragraph{Line~\ref{step:max}} 
We can find the maximum weight $W$ of an edge with high probability after 
$\tO(\sqrt{m})$ queries to the graph using the
quantum maximum finding routine of D{\"{u}}rr and H{\o}yer \cite{DH96}.

\paragraph{Line~\ref{step:sparsifier}} 
We run quantum algorithm of Apers and de Wolf \cite{AdW19} given in 
\cref{thm:AdW} to find an $\eps$-cut sparsifier $H$ of $G$ with 
$\tO(n/\eps^2)$ edges.  As all edge weights of $G$ are integral and 
the largest weight of an edge is $W$, by the ``moreover'' part of this 
theorem the edge weights of $H$ are integral and the largest weight of an 
edge of $H$ is $O(\eps^2 nW)$.  This step succeeds with high probability 
and takes $\tO(\sqrt{mn}/\eps)$ queries.  

\paragraph{Line~\ref{step:maxflow}} For this step we can run a classical 
maximum $\st$ flow algorithm to compute a max flow $F$ in $H$.  As the 
graph is explicitly known, this step requires no queries. 

\paragraph{Line~\ref{step:contract}} After step \ref{step:maxflow}, the 
graph $H'$ is explicitly known.  Thus we can compute a $3\eps n$-strong 
partition of $H'$ classically with no queries.

\paragraph{Line~\ref{step:learnedges}} 
Let $\Pcal = \{A_1, \ldots, A_t\}$ be the $3\eps n$-strong partition of $H'$ 
computed in the previous step.  The graph $G'$ is formed by contracting 
vertices in the same set of this partition. By \cref{clm:weight_upper} the 
graph $G'$ has $O(\eps n^2 W)$ edges, assuming all previous 
steps have succeeded.  Consider the concatenation of all the lists in the 
adjacency list representation of $G$.  This gives a vector of dimension $2m$ 
and defines an ordering of edges of $G$ where every edge appears twice.  
Define a string $x \in \{0,1\}^{2m}$ using the same ordering where 
$x(e) = 1$ if the endpoints of $e$ are in distinct sets of the 
partition $\Pcal$ and $x(e) = 0$ otherwise.  A query to a bit of $x$ can 
be answered with a single query to the adjacency list of $G$.  By 
\cref{thm:qsearch} via Grover search in time $\tO(n \sqrt{m\eps W})$ 
with high probability we can determine if $x$ has at most $O(\eps n^2 W)$ 
ones, and if so learn the positions of all these ones.   If the number of 
ones in $x$ is larger than $O(\eps n^2 W)$ then we abort.  
With high probability we do not abort and once we learn the positions of all 
the ones in $x$ we can then classically learn the weight of the corresponding 
edges in $G$ with $O(\eps n^2 W)$ more queries.  Then we have explicitly 
learned the graph $G'$.

Once we have an explicit description of $G'$ we can run a classical max flow 
algorithm to compute a maximum $\ab$ flow $F'$ in $G'$,
where $a,b \in \Pcal$ are the sets of the partition containing $s,t$ 
respectively.  This takes no queries.  The flow of $F'$ gives 
$\lambda_{ab}(G')$.  To also compute the shore of a minimum $\ab$ cut of 
$G'$ we consider the residual graph of the flow $F'$ and compute the 
connected component containing $a$.  This can also be done classically 
with no queries.

\paragraph{Correctness and total running time}
Each step individually succeeds at least with high probability and so with high 
probability all steps will be correct.  We can thus invoke \cref{thm:correct} 
to conclude that the algorithm is correct with high probability.

Queries are only made on Lines \ref{step:max},\ref{step:sparsifier}, and 
\ref{step:learnedges}, and the number of queries made on these lines is 
$\tO(\sqrt{m}), \tO(\sqrt{mn}/\eps)$ and $\tO(n \sqrt{m \eps W})$.  The term
$\tO(\sqrt{m})$ is low order, and the choice $\eps = (nW)^{-1/3}$ makes the 
remaining two terms $\tO(n^{5/6} W^{1/3})$.  Thus the total number of 
queries is $\tO(\sqrt{m} n^{5/6} W^{1/3})$.
\end{proof}

We have the following result for the adjacency matrix model.
\begin{corollary}
\label{cor:adj}
Let $G$ be a graph with $n$ vertices and $m$ edges where all edge weights are 
integral and the largest weight of an edge is $W$.  Given vertices $s,t$ and 
oracle access to the adjacency matrix of $G$, there is a quantum algorithm 
that computes $\lambda_{st}(G)$ and the shore of a minimum $\st$ cut of $G$ 
with high probability after $\tO(n^{11/6}W^{1/3})$ queries.
\end{corollary}

\begin{proof}
The result is trivial if $W \ge \sqrt{n}$, so we just need to consider the 
case $W < \sqrt{n}$.  We follow the same algorithm \cref{alg:stmincut} 
again choosing $\eps = (nW)^{-1/3}$.

Line \ref{step:max} can be done with $\tO(n)$ queries using the quantum 
maximum finding routine of D{\"{u}}rr and H{\o}yer \cite{DH96}.
In Line~\ref{step:sparsifier}, we run the adjacency matrix version of the 
sparsifier algorithm from \cite{AdW19} (quoted in \cref{thm:AdW}), which 
takes $\tO(n^{3/2}/\eps) = \tO(n^{11/6} W^{1/3})$ queries.
In Line~\ref{step:learnedges} we find the $O(\eps n^2 W)$ edges of $G'$ by 
applying \cref{thm:qsearch} over the $O(n^2)$ entries of the adjacency matrix, 
rather than the $O(m)$ entries of the adjacency list.
This takes $\tO(n^2 \sqrt{\eps W}) = \tO(n^{11/6} W^{1/3})$ queries.

The remaining steps are performed classically on graphs explicitly 
constructed by the algorithm and require no queries.
\end{proof}

\section{Randomized lower bound}
In this section we show that a randomized algorithm that computes $\lambda_{st}(G)$ with success probability at least $9/10$ on simple graphs $G$ with $m$ edges must make 
$\Omega(m)$ queries to the adjacency list of $G$ in the worst case.  This holds true even for just determining if $s$ and $t$ are connected in $G$, which we call the USTCON problem.
\begin{definition}[USTCON]
In the USTCON problem one is given adjacency list access to an undirected graph $G = (V,E)$ and two distinguished vertices $s,t \in V$.  The problem is to determine 
if $s$ and $t$ are connected in $G$.  
\end{definition}

To show a lower bound on USTCON we will use the classical adversary method for 
randomized query complexity \cite{Aaronson06, LM08, AKPVZ21}, a randomized 
analog of the quantum adversary method \cite{Ambainis02}.  We will just need a 
simple unweighted version of the method, adapted from 
\cite[Theorem 4.3]{Aaronson06}.

\begin{lemma}[{cf.\ \cite[Theorem 4.3]{Aaronson06}}]
\label{lem:adv_simple}
For a finite set $\Sigma$ and $S \subseteq \Sigma^n$, let $f : S \rightarrow \{0,1\}$.  Let $X \subseteq f^{-1}(0)$ and $Y \subseteq f^{-1}(1)$.  Let $R \subseteq X \times Y$ be such that 
for every $x \in X$ there are at least $t$ different $y \in Y$ such that $(x,y) \in R$ and for every $x \in X$ and $k \in [n]$ there are at most $\ell$ different $y \in Y$ such that $(x,y) \in R$ 
and $x_k \ne y_k$.  Then any randomized algorithm that computes $f$ with success probability $9/10$ must make $\Omega(t/\ell)$ queries.
\end{lemma}

\begin{theorem}
\label{thm:rand_lower}
A randomized algorithm that solves USTCON with success probability at least $9/10$ on graphs with $m$ edges edges requires $\Omega(m)$ queries.
\end{theorem}

\begin{proof}
First we show the lower bound in the dense case where $m = \Omega(n^2)$ using \cref{lem:adv_simple}.   We construct sets of negative instances $X$ and positive instances $Y$.  
Suppose that $n$ is even.  Excluding $s$ and $t$ we partition the remaining $n-2$ vertices into two sets 
$A$ and $B$ each of size $(n-2)/2$.  In all instances, $s$ will be connected to all vertices in $A$ and $t$ will be connected to all vertices in $B$.  $X$ will consist of a single negative instance $G$ 
where we additionally put a clique on the vertices in $A$ and a clique on the vertices in $B$ and no further edges.  Thus $s$ and $t$ will not be connected in $G$.  Note that apart from $s$ and $t$,
every vertex has degree $(n-2)/2$ in $G$.

For $Y$ we create a family of $2 \binom{n-2}{2}^2$ positive instances.  A positive instance will be labeled by $a=\{a_1, a_2\} \in A^{(2)}, b=\{b_1, b_2\} \in B^{(2)}$ and a bit $c \in \{0,1\}$.  
Associated to $a,b,c$ we construct a graph $G_{a,b,c}$ as follows.  Take the graph $G$ and remove the edges $a$ and $b$.  If $c=0$ then add edges 
$\{\min\{a_1,a_2\}, \min\{b_1,b_2\}\}$ and $\{\max\{a_1,a_2\}, \max\{b_1,b_2\}\}$ to form $G_{a,b,c}$.  Otherwise if $c=1$ then add edges 
$\{\min\{a_1,a_2\}, \max\{b_1,b_2\}\}$ and $\{\max\{a_1,a_2\}, \min\{b_1,b_2\}\}$ to form $G_{a,b,c}$.  In both cases all vertices in $A \cup B$ have degree $(n-2)/2$ in $G_{a,bc}$ 
and there will be two edges connecting $A$ and $B$ so $\lambda_{st}(G_{a,b,c}) = 2$.  This completes the description of the instances.  

The degree sequences of all instances are the same, thus we may assume this is known to the algorithm and the algorithm does not need to make any degree queries.  We thus focus 
on queries to the name of the $\ith$ neighbor of a vertex $v$.  For this it is important to specify the ordering of vertices given in the adjacency lists.  For all vertices we will 
use the same ordering in their list.  Let $k = (n-2)/2$ and label the vertices of $A$ as $a_1, a_2, \ldots, a_k$ and the vertices of $B$ as $b_1, b_2, \ldots, b_k$.  We use the ordering 
$s < t < a_1 < b_1 < \cdots < a_k < b_k$.  

We are now ready to show the lower bound using \cref{lem:adv_simple}.  We put $G$ in relation with all $2 \binom{n-2}{2}^2$ of the $G_{a,b,c}$.  
Now consider how many of the $G_{a,b,c}$ differ from $G$ in a specific location of the adjacency list, specifically consider the $\ith$ neighbor of 
a vertex $a_j \in A$.  In $G$ the $\ith$ neighbor of $a_j$ is
\begin{enumerate}
	\item $s$ if $i=1$.
	\item $a_{i-1}$ if $i \le j$.
	\item $a_{i}$ if $i > j$.
\end{enumerate}
The $\ith$ neighbor of $a_j$ in $G_{a,b,c}$ will be the same unless 
\begin{enumerate}
	\item $i \le j$ and $a=\{a_{i-1},a_j\}$.
	\item $i > j$ and $a = \{a_j, a_i\}$.
\end{enumerate} 
In either case, the number of $a,b,c$ for which $G_{a,b,c}$ differs from $G$ on the name of the $\ith$ neighbor of $a_j$ is $2 \binom{n-2}{2}$, as one only has free choice of 
the edge $b$ and the bit $c$.  A similar argument holds when considering the $\ith$ neighbor of a vertex $b_j \in B$.  Thus for any position of the adjacency list the number of 
$G_{a,b,c}$ that differ from $G$ is at most $2 \binom{n-2}{2}$ and by \cref{lem:adv_simple} we obtain a lower bound of $\binom{n-2}{2} = \Omega(m)$, proving the theorem 
in the dense case.

Finally, let us treat the general case of graphs with at most $m$ edges.  We choose the largest integer $p$ such that $2(p + \binom{p}{2}) \le m$ and then take disjoint sets of vertices 
$A$ and $B$ both of size $p$.  We then repeat the construction from the dense case on $s,t,A,B$.  The lower bound will be $\binom{p}{2} = \Omega(m)$ as desired.
\end{proof}

\section*{Acknowledgements}
We would like to thank an anonymous reviewer from ICALP 2022 for pointing out
errors in a previous version of this paper.   
Troy Lee is supported in part by the Australian Research Council Grant No: DP200100950.
Simon Apers is partially supported by French projects EPIQ (ANR-22-PETQ-0007), 
QUDATA (ANR-18-CE47-0010) and QUOPS (ANR22-CE47-0003-01), and EU project QOPT 
(QuantERA ERA-NET Cofund 2022-25)

\newcommand{\etalchar}[1]{$^{#1}$}

\end{document}